\documentclass[a4,12pt]{amsart}
\usepackage{amssymb}

\theoremstyle{plain}
\newtheorem{theorem}{Theorem}
\newtheorem{proposition}[theorem]{Proposition}

\theoremstyle{definition}
\newtheorem*{example}{Example}
\newtheorem*{remark}{Remark}
\newtheorem{definition}{Definition}
\newcommand{\N}{\mathcal{N}}
\newcommand{\Hi}{\mathcal{H}}
\newcommand{\Hii}{\mathcal{H}_0}
\newcommand{\al}{\alpha}
\newcommand{\f}{\varphi}
\newcommand{\la}{\lambda}
\newcommand{\La}{\Lambda}
\renewcommand{\leq}{\leqslant}
\renewcommand{\geq}{\geqslant}

\begin{document}
\title{Weakly sufficient quantum statistics}
\author{Katarzyna Lubnauer, Andrzej \L uczak, Hanna Pods\c{e}dkowska}
\address{Faculty of Mathematics and Computer Science\\
        \L\'od\'z University \\
        ul. S. Banacha 22 \\
        90-238 \L\'od\'z, Poland}
\email[Katarzyna Lubnauer]{lubnauer@math.uni.lodz.pl}

\email[Andrzej \L uczak]{anluczak@math.uni.lodz.pl}

\email[Hanna Pods\c{e}dkowska]{hpodsedk@math.uni.lodz.pl}

\keywords{Quantum statistic, weak sufficiency, von Neumann algebra}
\subjclass{Primary: 46L53; Secondary: 81S05}
\date{}
\begin{abstract}
  Some aspects of weak sufficiency of quantum statistics are investigated. In particular, we give necessary and sufficient conditions for the existence of a weakly sufficient statistic for a given family of vector states, investigate the problem of its minimality, and find the relation between weak sufficiency and other notions of sufficiency employed so far.
\end{abstract}
\maketitle

\section{Introduction}
Weak sufficiency of a quantum statistic was introduced in \cite{C} (under the name of `sufficiency'), and afterwards analyzed in \cite{LLP}. The definition of weak sufficiency is closely related to the classical factorization criterion, and seems to be especially well motivated in the case when we are dealing with the full algebra of bounded operators on a Hilbert space together with vector states. In the present paper, which can be considered as a follow-up to \cite{LLP}, we continue the investigation of this notion.  Three questions are dealt with: the problem of the existence of a weakly sufficient statistic for a given family of states, minimality of weakly sufficient statistics, and the relation of weak sufficiency to other notions of sufficiency. In particular, we show that, essentially, weak sufficiency follows from any of the notions of sufficiency employed so far.

\section{Preliminaries, notation and the quantum setup}
Let $\mathcal{H}$ be a Hilbert space with a scalar product $\langle\cdot,\cdot\rangle$. For an orthogonal projection $p$ we set
$p^{\bot}=\boldsymbol{1}-p$, where $\boldsymbol{1}$ is the identity
operator on $\mathcal{H}$. For $\xi\in\mathcal{H}$ we shall denote by $P_{[\xi]}$ the orthogonal projection onto the subspace spanned by $\xi$. In what follows the word \emph{projection} will always mean \emph{orthogonal projection}. $\mathbb{B}(\mathcal{H})$ will stand
for the algebra of all bounded linear operators on $\mathcal{H}$.

By a von Neumann algebra $\mathcal{M}$ of operators acting on
$\mathcal{H}$ we mean a $*$-algebra
$\mathcal{M}\subset\mathbb{B}(\mathcal{H})$ which is closed in the
strong operator topology on $\mathbb{B}(\mathcal{H})$, i.e., the
topology given by the family of seminorms
\[
 \mathbb{B}(\mathcal{H})\ni x\mapsto\|x\xi\|,\qquad
 \xi\in\mathcal{H}.
\]

For a von Neumann algebra of operators $\mathcal{M}$ acting on a
Hilbert space $\mathcal{H}$ we denote by $\mathcal{M}'$ the
\emph{commutant} of $\mathcal{M}$, i.e., the algebra of all bounded
operators on $\mathcal{H}$ which commute with all the operators
from $\mathcal{M}$. In particular, if $\mathcal{M}$ is abelian
then $\mathcal{M}\subset\mathcal{M}'$.

An abelian von Neumann algebra $\mathcal{M}$ is called \emph{maximal
abelian} if $\mathcal{M}=\mathcal{M}'$.

Let $p'$ be a projection in $\mathcal{M}'$. Then we can consider the
operators $xp'$ with $x\in\mathcal{M}$, restricted to the space
$p'\mathcal{H}$. The von Neumann algebra
\[
 \mathcal{M}_{p'}=\{xp'|p'\mathcal{H}:x\in\mathcal{M}\}
\]
of operators acting on the Hilbert space $p'\mathcal{H}$ is called
an \emph{induced von Neumann algebra}.

We shall not use any advanced theory of von Neumann algebras; some
basic necessary facts can be found for instance in
\cite{KR1,KR2,T}.

The $\sigma$-field of Borel subsets of the real line $\mathbb{R}$
will be denoted by $\mathcal{B}(\mathbb{R})$.

The most basic `probability space' employed to describe a quantum
system consists of a separable Hilbert space $\mathcal{H}$ and the
so-called pure (or vector) state represented by a unit vector
$\varphi\in\mathcal{H}$. A `noncommutative (or quantum) random
variable', called usually \emph{observable}, is a selfadjoint operator
$T$ on $\mathcal{H}$, so that for the spectral decomposition
\begin{equation}\label{T}
 T=\int_{-\infty}^{\infty}\lambda\,e(d\lambda)
\end{equation}
of $T$, where
$e\colon\mathcal{B}(\mathbb{R})\to\mathcal{P}(\mathcal{H})$ is the
spectral ($\equiv$ projection-valued) measure of $T$, the quantity
\[
 \langle e(E)\varphi,\varphi\rangle=\|e(E)\varphi\|^2,\qquad
 E\in\mathcal{B}(\mathbb{R}),
\]
represents the probability that being in the state $\varphi$ the
observable $T$ takes value in the set $E$.

The function $\mu_{\varphi}(\cdot)=\langle
e(\cdot)\varphi,\varphi\rangle$ is a genuine probability
distribution on the space $(\mathbb{R},\,\mathcal{B}(\mathbb{R}))$
of all possible values of the observable $T$. In quantum
statistics we are dealing with a family
$\{\varphi_{\theta}:\theta\in\Theta\}$ of possible states from
which we want to pick up the true state of our physical system.
Consequently, a `quantum statistical space' in our case is
$(\mathcal{H},\{\varphi_{\theta}:\theta\in\Theta\})$. As in the
classical case, various observables (i.e., `quantum random variables') prove
useful in obtaining information about this true state, so we shall
use for them the name of (quantum) \emph{statistics}.

A family of vector states $\{\f_{\theta}:\theta\in\Theta\}$ is said to be \emph{faithful} if for each $x\geq 0$ from $\mathbb{B}(\Hi)$ the equality $\langle x\f_{\theta},\f_{\theta}\rangle=0$ for all $\theta\in\Theta$ implies $x=0$. It is easily seen that $\{\f_{\theta}:\theta\in\Theta\}$ is faithful if and only if $[\{\f_{\theta}:\theta\in\Theta\}]=\Hi$, where for $\mathcal{K}\subset\mathcal{H}$, $[\mathcal{K}]$ stands for
the smallest closed linear subspace of $\mathcal{H}$ containing
$\mathcal{K}$.

Let $\mathcal{N}$ be the von Neumann algebra
generated by a quantum statistic $T$. Then for $T$ having the
spectral decomposition given by equation \eqref{T} we have
\begin{align*}
 \mathcal{N}=\Big\{\Phi(T)=&\int_{-\infty}^{\infty}\Phi(\lambda)
 \,e(d\lambda):\\
 &\Phi\text{ --- a complex-valued bounded Borel function}\Big\}.
\end{align*}
Thus $\mathcal{N}$ is an abelian von Neumann algebra determined by the spectral measure $e$. Note that the statistic $T$ itself needn't belong to the algebra $\mathcal{N}$ since the function $\Phi(\lambda)=\lambda$ is not bounded. We shall be concerned with statistics of the form
\[
 \Phi(T)=\int_{-\infty}^{\infty}\Phi(\lambda)\,e(d\lambda),
\]
where $\Phi$ is a real-valued Borel function. Such statistics, which are selfadjoint operators, are said to be \emph{affiliated with} $\mathcal{N}$.

The most general notion of sufficiency was introduced in \cite{P1,P2} (and further studied in \cite{JP}) as follows. Let $\mathcal{M}$ be a von Neumann algebra and let $\mathcal{N}$ be its von Neumann subalgebra. $\mathcal{N}$ is said to be \emph{sufficient} for a family of (arbitrary, not necessarily vector) states $\{\varphi_{\theta}:\theta\in\Theta\}$ if there exists a two-positive normal unital map $\alpha$ from $\mathcal{M}$ into $\mathcal{N}$ such that
\[
 \varphi_{\theta}\circ\alpha=\varphi_{\theta}
 \qquad\text{for all }\,\theta\in\Theta.
\]
In particular, if $\alpha$ is a conditional expectation then we obtain the notion of sufficiency introduced by H. Umegaki (see \cite{U1,U2}).

In our case the algebra $\mathcal{N}$ is the abelian algebra generated by a selfadjoint operator $T$ (or equivalently by its spectral measure), \linebreak $\mathcal{M}=\mathbb{B}(\mathcal{H})$, and the states $\varphi_{\theta}$ are vector states.

The definition of weak sufficiency we adopt was introduced in \cite{C} and consists in, roughly speaking, the possibility of obtaining the vectors $\varphi_{\theta}$ by applying functions of $T$ to a given vector $\chi\in\mathcal{H}$, in accordance with the classical factorization criterion. However, since for any vector state $\varphi$ and any complex number $c$ of modulus one the vector $c\varphi$ determines the same state as $\varphi$, in the definition of weak sufficiency we should take this fact into account. To put it in a simple way, let us agree to call a vector $\widetilde{\varphi}$ a \emph{version} of the vector $\varphi$ if there is a complex number $c$ of modulus one such that
\[
 \widetilde{\varphi}=c\varphi.
\]
Now the definition of weak sufficiency reads
\begin{definition}
 A statistic $T$ is \emph{weakly sufficient} for a family of vector
 states $\{\varphi_{\theta}:\theta\in\Theta\}$ if there exist Borel functions
 $\Phi_{\theta}\colon\mathbb{R}\to\mathbb{R}$, a unit vector
 $\chi$ in $\mathcal{H}$, and a version $\widetilde{\varphi}_{\theta}
 \text{ of }\varphi_{\theta}$ for each $\theta\in\Theta$, such that
 \begin{equation*}
  \widetilde{\varphi}_{\theta}=\Phi_{\theta}(T)\chi\qquad\text{for
  all }\,\theta\in\Theta.
 \end{equation*}
\end{definition}
Obviously, the requirement that $\|\chi\|=1$ is inessential, and
we shall omit it in the sequel.

For later use we state here two main results from \cite{LLP}.
\begin{theorem}[general case]\label{M1}
Let $T$ be a quantum statistic on a separable Hilbert space $\mathcal{H}$ with
the spectral decomposition given by \eqref{T}, and let
$\{\varphi_{\theta}:\theta\in\Theta\}$ be a family of vector states in $\mathcal{H}$. Denote, as before, by $\mathcal{N}$ the von Neumann algebra generated by $T$. The following conditions are equivalent:
\begin{enumerate}
 \item[(i)] $T$ is weakly sufficient for $\{\varphi_{\theta}:\theta\in
 \Theta\}$.
 \item[(ii)] There is a projection $p'$ in $\mathcal{N}'$ such
 that $\{\varphi_{\theta}:\theta\in\Theta\}\subset p'\mathcal{H}$ and the
 induced algebra $\mathcal{N}_{p'}$ is maximal abelian, and for each
 $\theta\in\Theta$ there is a version $\widetilde{\varphi}_{\theta}
 \text{ of }\varphi_{\theta}$ such that for any Borel set
 $E\subset\mathbb{R}$ and any  $\theta',\theta''\in\Theta$ we have
 \begin{equation}\label{rzecz1}
 \langle e(E)\widetilde{\varphi}_{\theta'},\widetilde{\varphi}_{\theta''}
 \rangle\in\mathbb{R}.
 \end{equation}
\end{enumerate}
\end{theorem}
\begin{theorem}[discrete case]\label{M2}
Let $T$ be a quantum statistic on a separable Hilbert space
$\mathcal{H}$ with the spectral decomposition
\begin{equation}\label{drs}
 T=\sum_k\lambda_ke_k,
\end{equation}
where $\{e_k\}$ is a countable partition of the identity, and
$\lambda_k\in\mathbb{R}$ are different eigenvalues of $T$,  and let
$\{\varphi_{\theta}:\theta\in\Theta\}$ be a family of vector states in $\mathcal{H}$. The following conditions are equivalent:
\begin{enumerate}
 \item[(i)] $T$ is weakly sufficient for $\{\varphi_{\theta}:\theta\in
 \Theta\}$.
 \item[(ii)]For each $k,\: \dim[\{e_k\varphi_{\theta}:\theta\in\Theta\}]
 =0\text{ or }1$, and for each $\theta\in\Theta$ there is a version
 $\widetilde{\varphi}_{\theta}\text{ of }\varphi_{\theta}$ such that
  for each $k$ and any\\  $\theta',\theta''\in\Theta,$ we have
 $\langle e_k\widetilde{\varphi}_{\theta'},\widetilde{\varphi}_{\theta''}
 \rangle\in\mathbb{R}$.
\end{enumerate}
\end{theorem}

\section{Existence and minimality of weakly sufficient quantum statistics}
We begin our analysis from the problem of the existence of a sufficient quantum statistic. The theorem below was proved in \cite{C} for a finite number of states.
\begin{theorem}\label{T1}
Let $\{\varphi_{\theta}:\theta\in\Theta\}$ be an arbitrary family of vector states. There exists a weakly sufficient statistic for this family if and only if for each $\theta\in\Theta$ there is a version $\widetilde{\varphi}_{\theta}$ of $\varphi_{\theta}$ such that for any $\theta',\theta''\in\Theta$
\begin{equation}\label{real}
 \langle\widetilde{\varphi}_{\theta'},\widetilde{\varphi}_{\theta''}\rangle
 \in\mathbb{R}.
\end{equation}
\end{theorem}
\begin{proof}
Assume that for each $\theta$ we have a version $\widetilde{\varphi}_{\theta}$ of $\varphi_{\theta}$ such that condition \eqref{real} holds. Since $\mathcal{H}$ is separable there is a countable subset $\{\widetilde{\varphi}_{\theta_n}:n=1,2,\dots\}\text{ of
}\{\widetilde{\varphi}_{\theta}:\theta\in\Theta\}$ such that
\[
 [\{\widetilde{\varphi}_{\theta_n}:n=1,2,\dots\}]
 =[\{\widetilde{\varphi}_{\theta}:\theta\in\Theta\}].
\]
Certainly, we may assume that $\{\widetilde{\varphi}_{\theta_n}\}$ are linearly independent. Let $\{\xi_n\}$ be the Gramm-Schmidt orthogonalization of $\{\widetilde{\varphi}_{\theta_n}\}$. Then
\[
 \xi_n=\sum_{j=1}^n \gamma_j^{(n)}\widetilde{\varphi}_{\theta_j},
\]
and condition \eqref{real} implies that $\gamma_j^{(n)}$ are real for all $j=1,\dots,n,\\n=1,2,\dots$. Put
\[
 T=\sum_{n=1}^{\infty}\la_nP_{[\xi_n]},
\]
where $\la_n$ are arbitrary different real numbers. We have
\[
 \dim[\{P_{[\xi_n]}\widetilde{\varphi}_{\theta}:\theta\in\Theta\}]=0
 \text{ or }1,
\]
since $P_{[\xi_n]}$ are one-dimensional projections, and for each $\theta',\theta''\in\Theta$
\begin{align*}
 &\langle P_{[\xi_n]}\widetilde{\varphi}_{\theta'},
 \widetilde{\varphi}_{\theta''}\rangle=
 \langle\widetilde{\varphi}_{\theta'},\xi_n\rangle\langle\xi_n,
 \widetilde{\varphi}_{\theta''}\rangle\\=&\Big\langle\widetilde{\varphi}_{\theta'},
 \sum_{j=1}^n \gamma_j^{(n)}\widetilde{\varphi}_{\theta_j}\Big\rangle
 \Big\langle\sum_{k=1}^n \gamma_k^{(n)}\widetilde{\varphi}_{\theta_k},
 \widetilde{\varphi}_{\theta''}\Big\rangle\\=&\sum_{j=1}^n\sum_{k=1}^n\gamma_j^{(n)}
 \gamma_k^{(n)}\langle\widetilde{\varphi}_{\theta'},\widetilde{\varphi}_{\theta_j}
 \rangle\langle\widetilde{\varphi}_{\theta_k},\widetilde{\varphi}_{\theta''}\rangle
 \in\mathbb{R},
\end{align*}
on account of condition \eqref{real}. From Theorem \ref{M2} we obtain that $T$ is weakly sufficient for $\{\varphi_{\theta}:\theta\in\Theta\}$.

Let now $T$ be a weakly sufficient statistic for $\{\varphi_{\theta}:\theta\in\Theta\}$. From Theorem \ref{M1} we get that for each $\theta\in\Theta$ there is a version $\widetilde{\varphi}_{\theta}$ of $\varphi_{\theta}$ such that for any Borel set $E\subset\mathbb{R}$ and any $\theta',\theta''\in\Theta$ we have
\[
 \langle e(E)\widetilde{\varphi}_{\theta'},
 \widetilde{\varphi}_{\theta''}\rangle\in\mathbb{R},
\]
where $e$ is the spectral measure of $T$. Putting $E=\mathbb{R}$ we obtain condition \eqref{real}.
\end{proof}
Now we want to address the question of minimality of a weakly sufficient statistic. Let us recall that in the classical case a sufficient statistic $S$ is said to be minimal if it is a function of any other sufficient statistic. In our situation this definition is too general and doesn't make much sense. Indeed, Theorem \ref{T1} shows that for a given family of states $\{\varphi_{\theta}:\theta\in\Theta\}$ satisfying condition \eqref{real} there are many weakly sufficient statistics having different spectral decompositions, while the classical definition of minimality of $S$ would require that $S=\Phi(T)$ for any weakly sufficient statistic $T$, i.e., roughly speaking, the spectral measure of $S$ should be a function of the spectral measure of $T$. However, it makes sense to speak about minimality with respect to the algebra $\mathcal{N}$ generated by a weakly sufficient statistic $T$. In the following definition the phrase ``weakly sufficient'' means ``weakly sufficient for a given family of vector states $\{\varphi_{\theta}:\theta\in\Theta\}$''.
\begin{definition}
Let $T$ be a weakly sufficient statistic. A weakly sufficient statistic $S$ is said to be \emph{minimal with respect to} $T$ if for any weakly sufficient statistic $U$ affiliated with the von Neumann algebra $\mathcal{N}$ generated by $T$, there is a real-valued Borel function $\Psi$ such that $S=\Psi(U)$.
\end{definition}
In what follows we shall investigate in detail the question of the existence of a minimal statistic with respect to a given weakly sufficient statistic $T$ in the case when $T$ has discrete spectral decomposition. Thus assume that $T$ is given by equation \eqref{drs}, where $\{e_k\}$ is a countable partition of the identity, and
$\lambda_k\in\mathbb{R}$ are different eigenvalues of $T$. To fix attention, assume that $k$ runs over all positive integers $\mathbb{N}$.  Let $\Phi$ be an arbitrary Borel function. Then
\[
 \Phi(T)=\sum_{k=1}^{\infty}\Phi(\lambda_k)e_k.
\]
The set $\mathbb{N}$ can be divided into nonempty disjoint sets $\mathbb{I}_k$ such that for $j,m\in\mathbb{I}_k$ we have
\[
 \Phi(\lambda_j)=\Phi(\lambda_m).
\]
Putting
\[
 \beta_k=\Phi(\lambda_j)\quad\text{for}\quad j\in\mathbb{I}_k,
\]
\begin{equation}\label{f}
 f_k=\sum_{j\in\mathbb{I}_k}e_j,
\end{equation}
we obtain the representation
\begin{equation}\label{fcjaT}
 \Phi(T)=\sum_k\beta_kf_k,
\end{equation}
where $\{f_k\}$ is a partition of the identity, and $\beta_k$ are different eigenvalues of $\Phi(T)$ ($\Phi(T)$ is often called \emph{coarse-graining} of $T$).

Let $T$ as above be a weakly sufficient statistic for a family of vector states $\{\varphi_{\theta}:\theta\in\Theta\}$. Condition (i) of Theorem \ref{M2} can be rewritten in the form
\begin{enumerate}
 \item[(i${}'$)] for each $k$ there exist a function $\theta\mapsto\gamma_k(\theta)$ and a unit vector $\xi_k$ such that
     \[
      e_k\varphi_{\theta}=\gamma_k(\theta)\xi_k\qquad\text{for all}\quad\theta\in\Theta.
     \]
 \end{enumerate}
Moreover, since $e_k$ are orthogonal we have $\xi_j\perp\xi_m$ for $j\ne m$.

With the help of condition (i${}'$) we define in $\mathbb{N}$ an equivalence relation $\sim$\, as follows:
\[
 j\sim m \quad\text{if there exists }\beta\ne0\;\text{such that }
 \gamma_j(\theta)=\beta\gamma_m(\theta)\text{ for all }\theta\in\Theta.
\]
The proposition below characterizes weakly sufficient statistics affiliated with the von Neumann algebra $\mathcal{N}$.
\begin{proposition}\label{P1}
Let $\Phi$ be a real-valued function, and let $U=\Phi(T)$ has form \eqref{fcjaT}. The following conditions are equivalent
\begin{enumerate}
 \item[(i)] $U$ is weakly sufficient
 \item[(ii)] for each $k\in\mathbb{N}$ and for any $j,m\in\mathbb{I}_k$ such that $\gamma_j$ and $\gamma_m$, as defined before, are non-zero functions, we have $j\sim m$.
\end{enumerate}
\end{proposition}
\begin{proof}
First, notice that for each $\theta\in\Theta$ we have
\[
 f_k\varphi_{\theta}=\sum_{j\in\mathbb{I}_k}e_j\varphi_{\theta}= \sum_{j\in\mathbb{I}_k}\gamma_j(\theta)\xi_j.
\]

Assume that (i) holds. Then according to Theorem \ref{M2} for each $k$
\[
 \dim[\{f_k\varphi_{\theta}:\theta\in\Theta\}]=0\text{ or }1,
\]
i.e., for each fixed $k$ we have
\[
 f_k\varphi_{\theta}=\sum_{j\in\mathbb{I}_k}\gamma_j(\theta)\xi_j
 =\gamma(\theta)\eta,
\]
for some function $\gamma$ and a unit vector $\eta$. Consequently, for $j,m\in\mathbb{I}_k$,
\[
 \gamma_j(\theta)=\gamma(\theta)\langle\eta,\xi_j\rangle,\qquad
 \gamma_m(\theta)=\gamma(\theta)\langle\eta,\xi_m\rangle.
\]
If the functions $\gamma_j,\gamma_m$ are non-zero then
\[
 \langle\eta,\xi_j\rangle\ne0,\qquad\langle\eta,\xi_m\rangle\ne0,
\]
and we obtain
\[
 \gamma_j(\theta)=\frac{\langle\eta,\xi_j\rangle}{\langle\eta,\xi_m\rangle}
 \gamma_m(\theta),
\]
showing that $j\sim m$.

Now let (ii) hold. Then for each $k$
\[
 f_k\varphi_{\theta}=\sum_j\gamma_j(\theta)\xi_j,
\]
where the sum is taken over these $j$ in $\mathbb{I}_k$ for which $\gamma_j$ is non-zero. From the equivalence of all these $j$ we obtain that there exist $j_0$ and numbers $\beta_j\ne0$ such that for each $j$ as above we have
\[
 \gamma_j(\theta)=\beta_j\gamma_{j_0}(\theta).
\]
Consequently, for each $\theta\in\Theta$
\[
 f_k\varphi_{\theta}=\sum_j\gamma_j(\theta)\xi_j=
 \sum_j\beta_j\gamma_{j_0}(\theta)\xi_j=
 \gamma_{j_0}(\theta)\sum_j\beta_j\xi_j=\gamma_{j_0}(\theta)\eta,
\]
where
\[
 \eta=\sum_j\beta_j\xi_j,
\]
showing that
\[
 \dim[\{f_k\varphi_{\theta}:\theta\in\Theta\}]=0\text{ or }1.
\]
Furthermore, since $T$ is weakly sufficient we have for each $k$
\[
 \langle e_k\widetilde{\varphi}_{\theta'},\widetilde{\varphi}_{\theta''}
 \rangle\in\mathbb{R}
\]
for some version $\widetilde{\varphi}_{\theta}$ of $\varphi_{\theta}$ and all $\theta',\theta''\in\Theta$, thus for this version
\[
 \langle f_k\widetilde{\varphi}_{\theta'},\widetilde{\varphi}_{\theta''}\rangle
 =\sum_{j\in\mathbb{I}_k}\langle e_j\widetilde{\varphi}_{\theta'},\widetilde{\varphi}_{\theta''}\rangle
 \in\mathbb{R},
\]
for each $k$ and all $\theta',\theta''\in\Theta$, which proves that $U$ is weakly sufficient.
\end{proof}
Now we can give a solution to the problem of the existence of a minimal weakly sufficient statistic with respect to a given weakly sufficient discrete statistic. Let us agree to call a family of vector states $\{\varphi_{\theta}:\theta\in\Theta\}$ \emph{nontrivial} if $\dim[\{\varphi_{\theta}:\theta\in\Theta\}]\geq2$.
\begin{theorem}
Let $T$ be a discrete statistic weakly sufficient for a nontrivial family of vector states $\{\varphi_{\theta}:\theta\in\Theta\}$, and assume that $T$ has form \eqref{drs}. There exists a weakly sufficient statistic minimal with respect to $T$ if and only if for each $k$
\begin{equation}\label{z}
 \dim[\{e_k\varphi_{\theta}:\theta\in\Theta\}]=1.
\end{equation}
\end{theorem}
\begin{proof}
Assume first that condition \eqref{z} holds. The equivalence relation $\sim$ divides the set $\mathbb{N}$ into nonempty pairwise disjoint sets $\mathbb{J}_m$ consisting of all mutually equivalent elements. Put
\[
 q_m=\sum_{j\in\mathbb{J}_m}e_j.
\]
Let $\varepsilon_m$ be arbitrary different real numbers, and let
\[
 S=\sum_m\varepsilon_m q_m.
\]
We shall show that $S$ is minimal. It is easily seen that $S$ is weakly sufficient. Let $U$ be an arbitrary weakly sufficient statistic affiliated with $\mathcal{N}$. Then $U=\Phi(T)$ for some real-valued Borel function $\Phi$, thus we may assume that $U$ has form \eqref{fcjaT}, with $f_k$ given by equation \eqref{f}. From assumption \eqref{z} and Proposition \ref{P1} it follows that for each $k$ all elements in $\mathbb{I}_k$ are equivalent, thus for each $k$ there is an $m$ such that $\mathbb{I}_k\subset\mathbb{J}_m$, so each $\mathbb{J}_m$ is a sum of some $\mathbb{I}_k$'s. Since
\[
 \bigcup_k\mathbb{I}_k=\bigcup_m\mathbb{J}_m=\mathbb{N},
\]
we obtain a partition of $\mathbb{N}$,
\[
 \mathbb{N}=\bigcup_m\mathbb{N}_m,
\]
into nonempty sets $\mathbb{N}_m$ such that for each $m$
\[
 \mathbb{J}_m=\bigcup_{k\in\mathbb{N}_m}\mathbb{I}_k.
\]
Now define a function $\Psi$ by the formula
\begin{align*}
 \Psi(\beta_k)&=\varepsilon_m  &\text{for}&\quad k\in\mathbb{N}_m,& &m=1,2,\dots\\
 \Psi(\la)&=0 &\text{for}&\quad\la\ne\varepsilon_m,& &m=1,2,\dots.
\end{align*}
It is easily seen that $S=\Psi(U)$, which shows that $S$ is minimal.

Assume now that there exists a minimal weakly sufficient statistic $S$, and that condition \eqref{z} fails. For the simplicity of notation let
\[
 \dim[\{e_1\varphi_{\theta}:\theta\in\Theta\}]=0.
\]
For $n\geqslant2$ define statistics $T_n$ by the formula
\[
 T_n=\lambda_n(e_1+e_n)+\sum_{k\notin\{1,\,n\}}\lambda_ke_k,
\]
where $\lambda_k$ and $e_k$ are as in \eqref{drs}. Since
\[
 (e_1+e_n)\varphi_{\theta}=e_n\varphi_{\theta},
\]
and $T$ is weakly sufficient, it follows that the $T_n$ are weakly sufficient too. From the minimality of $S$ there are real-valued Borel functions $\Psi,\,\Psi_n$, such that
\[
 S=\Psi(T)=\Psi_n(T_n)\qquad\text{for all}\quad n\in\mathbb{N}.
\]
We have
\[
 \Psi(T)=\sum_{k=1}^{\infty}\Psi(\lambda_k)e_k,\qquad
 \Psi_n(T_n)=\Psi_n(\lambda_n)(e_1+e_n)+\sum_{k\notin\{1,\,n\}}\Psi_n(\lambda_k)e_k,
\]
which implies that for each $n\in\mathbb{N}$
\[
 \Psi_n(\lambda_n)=\Psi(\lambda_1)=\Psi(\lambda_n).
\]
Consequently,
\[
 S=\Psi(T)=\sum_{k=1}^{\infty}\Psi(\lambda_k)e_k=
 \Psi(\lambda_1)\sum_{k=1}^{\infty}e_k=\Psi(\lambda_1)\boldsymbol{1},
\]
which contradicts the weak sufficiency of $S$.
\end{proof}

\section{Sufficiency vs weak sufficiency for quantum statistics}
In our further analysis we want to investigate in some detail the notion of sufficiency in the sense of Petz (cf. \cite{P1,P2} and \cite{JP}) adapted to the present setup, and to compare it with weak sufficiency. Since the algebra $\N$ is abelian, the two-positivity of the map $\al\colon\mathbb{B}(\Hi)\to\N$ defining sufficiency is equivalent to its positivity. Thus let us assume that $\al\colon\mathbb{B}(\Hi)\to\N$ is a normal positive contraction. Then for each $x\in\mathbb{B}(\Hi)$ there is a bounded Borel function $\Phi(\cdot;x)$ such that
\[
 \al(x)=\int_{-\infty}^{\infty}\Phi(\la;x)\,e(d\la).
\]
Moreover, we have
\[
 \underset{\la}{\sup\operatorname{ess}}|\Phi(\la;x)|=\|\al(x)\|\leqslant\|x\|.
\]
It turns out that $\al$ admits a more specific representation.
\begin{theorem}
Let $\{\varphi_{\theta}:\theta\in\Theta\}$ be a family of vector states, and let $\al\colon\mathbb{B}(\Hi)\to\N$ be a positive contraction such that the states $\f_{\theta}$ are $\al$-invariant. Then $\f_{\theta}$ are pairwise orthogonal, thus we have \linebreak $\{\varphi_{\theta}:\theta\in\Theta\}=\{\f_n\}$, where $n$ runs over a countable set of positive integers. Moreover, there exist pairwise orthogonal non-zero projections $e_n\in\N$ and a Borel set $E\subset\mathbb{R}$ such that
\[
 \sum_ne_n=e(E),
\]
and
\begin{equation}\label{przedstal}
 \al(x)=\sum_n\langle x\f_n,\f_n\rangle\,e_n
 +\int_{E'}\Phi(\la;x)\,e(d\la),\qquad x\in\mathbb{B}(\Hi).
\end{equation}
\end{theorem}
\begin{proof}
For a fixed orthonormal basis in  $\Hi$, let $\Hii$ be the
set of all finite linear combinations of elements of this basis with
`rational complex' coefficients, where by a `rational complex number' is meant a complex number of the form $\beta+\gamma i$ with $\beta$ and $\gamma$ rational numbers. Then $\Hii$ is a dense countable subset of $\Hi$, closed with respect to addition and multiplication by rational complex numbers. For arbitrary $\xi,\eta\in\Hi$ let $t_{\xi,\eta}$ be operators on $\Hi$ defined by the formula
\[
 t_{\xi,\eta}\zeta=\langle\zeta,\eta\rangle\xi,\qquad\zeta\in\Hi.
\]
It is obvious that for any
$\xi_1,\xi_2,\xi,\eta_1,\eta_2,\eta\in\Hi,\,\gamma\in\mathbb{C}$ we
have
\begin{align*}
 &t_{\xi_1+\xi_2,\eta}=t_{\xi_1,\eta}+t_{\xi_2,\eta}\\
 &t_{\xi,\eta_1+\eta_2}=t_{\xi,\eta_1}+t_{\xi,\eta_2}\\
 &t_{\gamma\xi,\eta}=\gamma t_{\xi,\eta}.
\end{align*}
Consequently, for any $\xi_1,\xi_2,\eta\in\Hii$ we get
\begin{align*}
 &\int_{-\infty}^{\infty}\Phi(\la;t_{\xi_1+\xi_2,\eta})\,e(d\la)=
 \al(t_{\xi_1+\xi_2,\eta}) =\al(t_{\xi_1,\eta})+\al(t_{\xi_2,\eta})\\
 =&\int_{-\infty}^{\infty}[\Phi(\la;t_{\xi_1,\eta})+\Phi(\la;t_{\xi_2,\eta})]\,e(d\la),
\end{align*}
and thus
\begin{equation}\label{add}
 \Phi(\la;t_{\xi_1+\xi_2,\eta})=\Phi(\la;t_{\xi_1,\eta})
 +\Phi(\la;t_{\xi_2,\eta}) \quad e-a.e.,
\end{equation}
i.e., there is a set $\La_{\xi_1,\xi_2,\eta}$ of full $e$--measure
such that for all $\la\in\La_{\xi_1,\xi_2,\eta}$ equality \eqref{add}
holds. Similarly with addition in the second position and
multiplication by a rational complex number, thus taking (countable!) intersection of all these sets of full $e$--measure we obtain a set $\La_1$ of full $e$--measure such that for all $\la\in\La_1$, all
$\xi_1,\xi_2,\eta_1,\eta_2\in\Hii$, and all rational complex numbers
$\gamma$ we have
\begin{equation}\label{w1}
 \begin{aligned}
 &\Phi(\la;t_{\xi_1+\xi_2,\eta_1})=\Phi(\la;t_{\xi_1,\eta_1})+
 \Phi(\la;t_{\xi_2,\eta_1}),\\
 &\Phi(\la;t_{\xi_1,\eta_1+\eta_2})=\Phi(\la;t_{\xi_1,\eta_1})+
 \Phi(\la;t_{\xi_1,\eta_2}),\\
 &\Phi(\la;\gamma t_{\xi_1,\eta_1})=\gamma\Phi(\la;t_{\xi_1,\eta_1}).
 \end{aligned}
\end{equation}
Further, for each $\xi\in\Hii$ we have $t_{\xi,\xi}\geqslant 0$, thus
\[
 0\leqslant\al(t_{\xi,\xi})=\int_{-\infty}^{\infty} \Phi(\la;t_{\xi,\xi})\,e(d\la),
\]
which means that
\[
 \Phi(\la;t_{\xi,\xi})\geqslant 0 \quad e-a.e.
\]
Consequently, there is a set $\La_{\xi}$ of full $e$--measure such
that for all $\la\in\La_{\xi}$
\begin{equation}\label{w2}
 \Phi(\la;t_{\xi,\xi})\geqslant 0.
\end{equation}
Putting
\[
 \La_2=\bigcap_{\xi\in\Hii}\La_{\xi},
\]
we obtain a set $\La_2$ of full $e$--measure such that for all
$\la\in\La_2$ and all $\xi\in\Hii$
\[
 \Phi(\la;t_{\xi,\xi})\geqslant 0.
\]
For any fixed $\xi,\eta\in\Hii$ there is a set $\La_{\xi,\eta}$ of
full $e$--measure such that for all $\la\in\La_{\xi,\eta}$ we have
\[
 |\Phi(\la;t_{\xi,\eta})|\leqslant\underset{\omega}{\sup\operatorname{ess}}
 |\Phi(\omega;t_{\xi,\eta})|\leqslant\|t_{\xi,\eta}\|=\|\xi\|\|\eta\|.
\]
Put
\[
 \La_3=\bigcap_{\xi,\eta\in\Hii}\La_{\xi,\eta}.
\]
Then $\La_3$ is of full $e$--measure, and we have for all
$\la\in\La_3$ and all $\xi,\eta\in\Hii$
\begin{equation}\label{w3}
 |\Phi(\la;t_{\xi,\eta})|\leqslant\|\xi\|\|\eta\|.
\end{equation}
Put
\[
 \La=\La_1\cap\La_2\cap\La_3.
\]
Then $\La$ is of full $e$--measure, and for all $\la\in\La$, all $\xi_1,\xi_2,\xi,\eta_1,\eta_2,\eta\in\Hii$, and all rational complex numbers $\gamma$ relations \eqref{w1}, \eqref{w2} and \eqref{w3} hold.

For $\la\in\La$ we define on
$\Hii\times\Hii$ a function $h(\la;\cdot,\cdot)$ by the formula
\[
 h(\la;\xi,\eta)=\Phi(\la;t_{\xi,\eta}).
\]
Obviously, for all $\xi,\eta\in\Hii$, $h(\cdot;\xi,\eta)$ is a Borel function.
$h(\la;\cdot,\cdot)$ is a sesquilinear form on $\Hii\times\Hii$ with
respect to multiplication by rational complex numbers; moreover, we
have
\[
 |h(\la;\xi,\eta)|\leqslant\|\xi\|\|\eta\|\qquad\text{and}
 \qquad h(\la;\xi,\xi)\geqslant0.
\]
Let $\Hii\ni\xi_n\to\xi,\;\Hii\ni\eta_n\to\eta$. Then for any
$\la\in\La$ we have
\begin{equation}\label{jedn}
 \begin{aligned}
 &|h(\la;\xi_n,\eta_n)-h(\la;\xi_m,\eta_m)|\leqslant|h(\la;\xi_n,\eta_n)
 -h(\la;\xi_n,\eta_m)|\\+&|h(\la;\xi_n,\eta_m)-h(\la;\xi_m,\eta_m)|
 =|h(\la;\xi_n,\eta_n-\eta_m)|\\+&|h(\la;\xi_n-\xi_m,\eta_m)|\\\leqslant
 &\|\xi_n\|\|\eta_n-\eta_m\|+\|\xi_n-\xi_m\|\|\eta_m\|
 \underset{n,m\to\infty}{\longrightarrow}0,
 \end{aligned}
\end{equation}
thus we can define for $\la\in\La$, a function
$\widetilde{h}(\la;\cdot,\cdot)$ on $\Hi\times\Hi$ by the formula
\[
 \widetilde{h}(\la;\xi,\eta)=\lim_{n\to\infty}h(\la;\xi_n,\eta_n),
\]
where $\xi_n,\;\eta_n$ are as before. The reasoning as in
\eqref{jedn} shows that $\widetilde{h}(\la;\xi,\eta)$ does not depend
on the approximating sequences $\{\xi_n\},\:\{\eta_n\}$. Again, we immediately notice that for all $\xi,\eta\in\Hi$, $\widetilde{h}(\cdot;\xi,\eta)$ is a Borel function. $\widetilde{h}(\la;\cdot,\cdot)$ is clearly a sesquilinear form on
$\Hi\times\Hi$, and
\[
 |\widetilde{h}(\la;\xi,\eta)|=\lim_{n\to\infty}|h(\la;\xi_n,\eta_n)|
 \leqslant\lim_{n\to\infty}\|\xi_n\|\|\eta_n\|=\|\xi\|\|\eta\|.
\]
Thus there exists an operator $a(\la)$ of norm $\leqslant1$, such that
\[
 \langle a(\la)\xi,\eta\rangle=\widetilde{h}(\la;\xi,\eta),
 \qquad\xi,\eta\in\Hi;
\]
moreover, the operator-valued function $\la\mapsto a(\la)$ is weakly Borel measurable. For $\xi,\eta\in\Hii$ we have
\begin{equation}\label{al}
 \langle a(\la)\xi,\eta\rangle=h(\la;\xi,\eta)=\Phi(\la;t_{\xi,\eta}).
\end{equation}
In particular, for $\xi\in\Hii$
\[
 \langle a(\la)\xi,\xi\rangle=\Phi(\la;t_{\xi,\xi})\geqslant 0,
\]
showing that
\begin{equation}\label{nier}
 0\leqslant a(\la)\leqslant\boldsymbol{1}.
\end{equation}

Let now $\f$ be an $\al$-invariant vector state. Define an operator $a$ (depending on $\varphi$) as a weak integral
\[
 a=\int_{-\infty}^{\infty} a(\la)\,\|e(d\la)\f\|^2,
\]
i.e.,
\[
 \langle a\xi,\eta\rangle=\int_{-\infty}^{\infty}\langle a(\la)\xi,\eta\rangle\,
 \|e(d\la)\f\|^2,\qquad\xi,\eta\in\Hi.
\]
For $\xi,\eta\in\Hii$ we have
\begin{align*}
 \langle &P_{[\f]}\xi,\eta\rangle=\langle\xi,\f\rangle\langle\f,\eta\rangle
 =\langle t_{\xi,\eta}\f,\f\rangle=\langle\al(t_{\xi,\eta})\f,\f\rangle\\
 =&\int_{-\infty}^{\infty} \Phi(\la;t_{\xi,\eta})\,\|e(d\la)\f\|^2=\int_{-\infty}^{\infty}\langle a(\la)\xi,\eta\rangle\,
 \|e(d\la)\f\|^2=\langle a\xi,\eta\rangle,
\end{align*}
which shows that
\[
 a=P_{[\f]}.
\]
Consequently,
\[
 1=\langle a\f,\f\rangle=\int_{-\infty}^{\infty}\langle a(\la)\f,\f\rangle\, \|e(d\la)\f\|^2,
\]
and since
\[
 \langle a(\la)\f,\f\rangle\leqslant 1,
\]
we get
\[
 \langle a(\la)\f,\f\rangle=1\qquad \|e(\cdot)\f\|^2-a.e.
\]
On account of relation \eqref{nier}, the above equality yields
\[
 a(\la)\f=\f \qquad \|e(\cdot)\f\|^2-a.e.
\]
For $\xi\perp\f$ we have
\[
 0=\langle a\xi,\xi\rangle=\int_{-\infty}^{\infty}\langle a(\la)\xi,\xi\rangle \|e(d\la)\f\|^2,
\]
hence
\[
 \langle a(\la)\xi,\xi\rangle=0 \qquad \|e(\cdot)\f\|^2-a.e.,
\]
and thus
\[
 a(\la)\xi=0 \qquad \|e(\cdot)\f\|^2-a.e.
\]
Consequently, we obtain
\[
 a(\la)=P_{[\f]} \qquad \|e(\cdot)\f\|^2-a.e.,
\]
thus there is a set $E_{\f}\in\mathcal{B}(\mathbb{R})$ with
\begin{equation}\label{rf}
 \|e(E_{\f})\f\|^2=1,
\end{equation}
such that for all $\la\in E_{\f}$
\[
 a(\la)=P_{[\f]}.
\]
Relation \eqref{al} now yields for all $\xi,\eta\in\Hii$ and all
$\la\in\La\cap E_{\f}$
\begin{equation}\label{const}
 \Phi(\la;t_{\xi,\eta})=\langle P_{[\f]}\xi,\eta\rangle
 =\langle t_{\xi,\eta}\f,\f\rangle.
\end{equation}
Observe that since $e(E_{\f})$ is a projection, from equality \eqref{rf} it follows that
\begin{equation}\label{rf1}
 e(E_{\f})\f=\f.
\end{equation}

Let now $\f$ and $\psi$ be two different $\al$-invariant vector states.
Then we have
\[
 \La\cap E_{\f}\cap E_{\psi}=\emptyset.
\]
Indeed, if $\la\in\La\cap E_{\f}\cap E_{\psi}$ then we would have for all $\xi,\eta\in\Hii$
\[
 \langle P_{[\f]}\xi,\eta\rangle=\Phi(\la;t_{\xi,\eta})
 =\langle P_{[\psi]}\xi,\eta\rangle,
\]
giving the equality $P_{[\f]}=P_{[\psi]}$. In particular, since $\La$ is of full $e$--measure we obtain by \eqref{rf1}
\[
 \langle\f,\psi\rangle=\langle e(E_{\f})\f,e(E_{\psi})\psi\rangle=\langle e(E_{\f}\cap\La)\f,e(E_{\psi}\cap\La)\psi\rangle=0,
\]
showing that different $\al$-invariant vector states are orthogonal.

Let now $\{\f_{\theta}:\theta\in\Theta\}$ be a family of
$\al$-invariant vector states. Since $\f_{\theta}$ are pairwise
orthogonal we have $\{\f_{\theta}\}=\{\f_{\theta_n}\}$. Denote
$\f_{\theta_n}=\f_n$, and let $E_{\f_n}$ be the sets as above. Put
\begin{equation}\label{p}
 E_n=E_{\f_n}\cap\Lambda,\quad e_n=e(E_n),\quad E=\bigcup_n E_n,
 \quad p=\sum_n e_n.
\end{equation}
Then $E_n$ are pairwise disjoint, so $e_n$ are orthogonal,
$P_{[\f_n]}\leqslant e_n$, and for $\xi,\eta\in\Hii$ we have on account of \eqref{const}
\begin{equation}\label{postal}
 \begin{aligned}
 \al(t_{\xi,\eta})&=\int_E \Phi(\la;t_{\xi,\eta})\,e(d\la)+
 \int_{E'}\Phi(\la;t_{\xi,\eta})\,e(d\la)\\
 &=\sum_n\int_{E_n}\Phi(\la;t_{\xi,\eta})\,e(d\la)+
 \int_{E'}\Phi(\la;t_{\xi,\eta})\,e(d\la)\\&=\sum_n\langle P_{[\f_n]}
 \xi, \eta\rangle\, e(E_n)+\int_{E'}\Phi(\la;t_{\xi,\eta})\,e(d\la)\\
 &=\sum_n\langle t_{\xi,\eta}\f_n,\f_n\rangle\, e_n+
 \int_{E'}\Phi(\la;t_{\xi,\eta})\,e(d\la).
 \end{aligned}
\end{equation}
From \eqref{postal} we get for each $n$
\[
 e_n\al(t_{\xi,\eta})=\langle t_{\xi,\eta}\f_n,\f_n\rangle\,e_n,
\]
and taking into account the continuity of $\al$ we obtain by approximation
\[
 e_n\al(x)=\langle x\f_n,\f_n\rangle\,e_n
\]
for each $x\in\mathbb{B}(\Hi)$. Consequently,
\[
 p\,\al(x)=\sum_n\langle x\f_n,\f_n\rangle\,e_n,
\]
so
\begin{equation*}
 \begin{aligned}
 \al(x)&=p\,\al(x)+p^{\bot}\al(x)\\&=\sum_n\langle x\f_n,\f_n\rangle\,e_n
 +p^{\bot}\int_{-\infty}^{\infty}\Phi(\la;x)\,e(d\la)\\&=\sum_n\langle x\f_n,\f_n\rangle\,e_n
 +\int_{E'}\Phi(\la;x)\,e(d\la).
 \end{aligned}
\end{equation*}
\end{proof}
\begin{remark}
It should be noted that representation \eqref{przedstal} is non-unique since the sets $E_n$ are non-unique, and consequently, the projections $e_n$ are non-unique. However, in any case we have $P_{[\f_n]}\leq e_n$ and $e_ne_m=0$ for $n\neq m$.
\end{remark}
It turns out that despite its non-uniqueness representation \eqref{przedstal} yields a number of
interesting consequences. First, denote by $\N_0$ the algebra
\[
 \N_0=\Big\{\sum_n\gamma_ne_n:\gamma_n\in\mathbb{C}\Big\},
\]
and define a map $\al_p\colon\mathbb{B}(\Hi)\to\N_0$ by the formula
\[
 \al_p(x)=p\,\al(x),
\]
where $p$ is given by \eqref{p}. Then $\al_p$ is a normal positive projection of norm one onto the algebra $\N_0$. In particular, if the family $\{\f_n\}$ of states is faithful, then it is an orthonormal basis in $\Hi$, so
\[
 \sum_nP_{[\f_n]}=\boldsymbol{1},
\]
yielding the equalities
\[
 e_n=P_{[\f_n]}.
\]
It follows that $\N_0=\N$. Indeed, we have $\N_0\subset\N$ and since
$\N_0$ is generated by the minimal projections $P_{[\f_n]}$, it
is maximal abelian (cf. \cite[Theorem 9.4.1]{KR2}). Hence
\[
 \N\subset\N'\subset\N'_0=\N_0.
\]
Further
\[
 \al(x)=\sum_n\langle x\f_n,\f_n\rangle P_{[\f_n]},
\]
thus $\al$ is a conditional expectation onto $\N$.

Next, it turns out that the notions of sufficiency and weak sufficiency are related in a natural way.
\begin{theorem}
Let a statistic $T$ be sufficient in the sense of Petz for a family of vector states $\{\f_{\theta}:\theta\in\Theta\}$, i.e., there exists a positive normal unital map $\al\colon\mathbb{B}(\Hi)\to\N$ such that
\[
 \langle\alpha(x)\varphi_{\theta},\varphi_{\theta}\rangle=\langle x\varphi_{\theta},\varphi_{\theta}\rangle
 \qquad\text{for all }\,x\in\mathbb{B}(\Hi),\,\theta\in\Theta.
\]
If $\al(\mathbb{B}(\Hi))=\N$, then $T$ is weakly sufficient.
\end{theorem}
\begin{remark}
Observe that the theorem gives a full answer to the problem of the relation between weak sufficiency and sufficiency in the sense of Umegaki because then $\al$ is a conditional expectation, and thus `onto'.
\end{remark}
\begin{proof}
Adopting the notation as before we have $\{\f_{\theta}\}=\{\f_n\}$, and formula \eqref{przedstal} for $\al$. From the assumption that $\al$ is `onto' it follows that $e_n$ are minimal projections in $\N$. Indeed, if $q\leqslant e_n$, for a projection $q$ in $\N$, then
\[
 q=\al(x)
\]
for some $x\in\mathbb{B}(\Hi)$, and we get
\[
 q=qe_n=\al(x)e_n=\langle x\f_n,\f_n\rangle\,e_n,
\]
thus $q=e_n$.

From the minimality of $e_n$ and the relation
$P_{[\f_n]}\leqslant e_n$ it follows that $P_{[\f_n]}\in\N'$. Indeed,
if $P_{[\f_n]}q\ne qP_{[\f_n]}$ for some projection $q\in\N$, then
$P_{[\f_n]}q\ne0$, so
\[
 0\ne P_{[\f_n]}qP_{[\f_n]}=P_{[\f_n]}e_nqe_nP_{[\f_n]},
\]
giving
\[
 e_nq=e_nqe_n\ne0,
\]
and since $e_nq$ is a projection in $\N$ such that $e_nq\leqslant
e_n$, we get $e_nq=e_n$ because $e_n$ is minimal. Hence
$P_{[\f_n]}q=P_{[\f_n]}=qP_{[\f_n]}$, a contradiction.

Put
\[
 p'=\sum_nP_{[\f_n]}.
\]
Then $p'$ is a projection in $\N'$. The elements of the algebra
$\N p'$ have the form
\[
 \al(x)p'=\sum_n\langle x\f_n,\f_n\rangle P_{[\f_n]},
\]
and since the von Neumann algebra $\N_{p'}$ is generated by the minimal projections $P_{[\f_n]}|p'\Hi$, it is maximal abelian
(cf. \cite[Theorem 9.4.1]{KR2}). Moreover, $\{\f_n\}\subset p'\Hi$, and for any $n\ne m$ and each $B\in\mathcal{B}(\mathbb{R})$ we have
\begin{align*}
 \langle e(B)\f_n,\f_m\rangle&=\langle e(B)e_n\f_n,e_m\f_m\rangle
 =\langle e_m e(B)e_n\f_n,\f_m\rangle\\&=\langle e(B)e_m e_n \f_n,\f_m\rangle=0,
\end{align*}
hence on account of Theorem \ref{M1} we obtain the weak sufficiency of $\N$.
\end{proof}
\begin{remark}
As seen from the above proof, the assumption that $\al$ is unital may be abandoned.
\end{remark}
The following simple example shows that weak sufficiency is indeed weaker than sufficiency in the sense of Petz.
\begin{example}
 Let $\mathcal{H}=\mathbb{C}^2$,
\[
 \varphi_1=\begin{pmatrix}
  1\\0
  \end{pmatrix},\qquad
  \varphi_2=\begin{pmatrix}
   \frac{1}{\sqrt2}\\
   \frac{1}{\sqrt2}\end{pmatrix},
\]
and let $T$ be represented in the standard basis of $\mathbb{C}^2$ as
\[
  T=\begin{bmatrix}
   1 & 0\\
   0 & -1
  \end{bmatrix}.
\]
Then
\[
 \mathcal{N}=\bigg\{\begin{bmatrix}a & 0\\0 & b\end{bmatrix}: a,\,b\in\mathbb{C}\bigg\}.
\]
$T$ is weakly sufficient for the family of states $\{\varphi_1,\,\varphi_2\}$. Indeed, taking for instance $\chi=\varphi_2$,
\[
 T_1=\begin{bmatrix}
   \sqrt2 & 0\\
   0 & 0
  \end{bmatrix},\qquad
  T_2=\begin{bmatrix}
   1 & 0\\
   0 & 1
  \end{bmatrix},
\]
we immediately verify that
\[
 T_1\chi=\varphi_1,\qquad T_2\chi=\varphi_2.
\]
Let
\[
 A=\begin{bmatrix}
   a_{11} & a_{12}\\
   a_{21} & a_{22}
  \end{bmatrix}
\]
be an arbitrary element of $\mathbb{B}(\mathbb{C}^2)$ represented in the standard basis. Then
\begin{align*}
 &\langle A\varphi_1,\varphi_1\rangle=a_{11},\\
 &\langle A\varphi_2,\varphi_2\rangle=\frac{a_{11}+a_{12}+a_{21}+a_{22}}{2}.
\end{align*}
Let now $\alpha$ be an arbitrary linear map from $\mathbb{B}(\mathbb{C}^2)$ into $\mathcal{N}$. Then $\alpha(A)$ has the form
\[
 \alpha(A)=\begin{bmatrix}
  a & 0\\
  0 & b
  \end{bmatrix},
\]
with some $a,\,b\in\mathbb{C}$. Assume that $\varphi_1$ and $\varphi_2$ are $\alpha$-invariant. We have
\[
 a=\langle\alpha(A)\varphi_1,\varphi_1\rangle=\langle A\varphi_1,\varphi_1\rangle=a_{11},
\]
and
\begin{align*}
 \frac{a+b}{2}&=\langle\alpha(A)\varphi_2,\varphi_2\rangle\\&=\langle A\varphi_2,\varphi_2\rangle=\frac{a_{11}+a_{12}+a_{21}+a_{22}}{2},
\end{align*}
which gives
\[
 a=a_{11},\qquad b=a_{12}+a_{21}+a_{22}.
\]
Consequently, $\alpha(A)$ must be of the form
\[
 \alpha(A)=\begin{bmatrix}
  a_{11} & 0\\
  0 & a_{12}+a_{21}+a_{22}
  \end{bmatrix},
\]
which implies, as is clearly seen, that $\alpha$ is not a positive map. Thus there does not exist a positive map $\alpha$ from $\mathbb{B}(\mathbb{C}^2)$ to $\mathcal{N}$ such that $\varphi_1$ and $\varphi_2$ are $\alpha$-invariant, showing that $T$ is not sufficient in the sense of Petz.
\end{example}

\end{document}